\documentclass[submission,copyright]{eptcs}

\providecommand{\event}{MSFP 2012}

\newcommand{\Set}{\mathbf{\mathsf{Set}}}
\newcommand{\func}[1]{#1}
\newcommand{\fF}{\func{F}}
\newcommand{\fG}{\func{G}}
\newcommand{\fT}{\func{T}}

\newcommand{\larr}{\rightarrow }
\newcommand{\fun}[2]{#2^{#1}}
\newcommand{\x}{\times}

\newcommand{\cA}{\mathscr{A}}
\newcommand{\Id}{\func{\mathbf{1}}} 
\newcommand{\iid}{\mathsf{id}}
\newcommand{\Nat}{\mathbb{N}}
\newcommand{\Ext}{\mathsf{Ext}}

\newcommand{\flatten}{\mathsf{flat}} 

\usepackage{amsmath,amsthm,mathrsfs}
\usepackage{enumerate}

\theoremstyle{plain}
\newtheorem{theorem}{Theorem}[section]
\newtheorem{lemma}[theorem]{Lemma}

\newtheorem{definition}[theorem]{Definition}

\theoremstyle{definition}

\newtheorem{example}[theorem]{Example}

\theoremstyle{remark}
\newtheorem{remark}[theorem]{Remark}

\newcommand{\oast}{\circledast}

\usepackage[all]{xy}

\author{Mauro Jaskelioff
\institute{CIFASIS, Rosario, Argentina}
\institute{FCEIA, Universidad Nacional de Rosario, Argentina} 
\and 
Ondrej Rypacek
\institute{King's College, London, UK}}
\title{An Investigation of the Laws of Traversals}

\def\titlerunning{An Investigation of the Laws of Traversals}
\def\authorrunning{M. Jaskelioff \& O. Rypacek}

\input{lhs2tex.tex}

\begin{document}
\maketitle
\begin{abstract}
  Traversals of data structures are ubiquitous in
  programming. Consequently, it is important to be able to
  characterise those structures that are traversable and understand
  their algebraic properties. Traversable functors have been
  characterised by McBride and Paterson as those equipped with a
  distributive law over arbitrary applicative functors; however, laws
  that fully capture the intuition behind traversals are missing. This
  article is an attempt to remedy this situation by proposing laws for
  characterising traversals that capture the intuition behind them. To
  support our claims, we prove that finitary containers are
  traversable in our sense and argue that elements in a traversable
  structure are visited exactly once.

\end{abstract}

\section{Introduction}
\label{sec:introduction}
Traversals of data structures are ubiquitous in
programming. Consequently, it is essential for the writer of
mathematically structured programs to have a precise understanding of
the abstract structure of traversals and of their algebraic properties.

There are many notions of traversal, but in this article, we will
focus on the traversals of functors $G : \Set\to\Set$ as given by a
family of natural transformations $\delta^F: GF \to FG$ over
applicative functors
$F$~\cite{mcbride08:applicative-programming}. This notion of traversal
is quite abstract but practical for structuring and reasoning about
programs~\cite{Gibbons.Oliveira.Iterator,DJPTFP:2011,backwards}, and
encompasses other notions of traversals, such as generating a list of
elements.

Usually, distributive laws are required to respect structure via
coherence laws~\cite{Rypacek:2010}. However, the characterisation of
traversability in terms of a distributive law given by McBride and
Paterson is incomplete as no coherence laws are required to hold and
one can write distributive laws that do not follow the intuition of
what a traversal should be. Consequently, a distributive law is not
enough and the definition of a traversal needs to be strengthened to
avoid bad instances. However, there seems to be no consensus of
precisely what a traversal should be.

Moggi et al.~\cite{MoggiBJ99} define traversals of a functor $G$ over
a functor $F$ to be a family over $G1$ of distributive laws
$\delta^F_{s:G1} : G_sF\to FG_s$. Here, $G_s$ is a family over $G1$ of
functors obtained from the pullback of the diagram $\xymatrix@1{1
  \ar[r]^s & G1 & \ar[l]_{G!}  GX}$. The family of morphisms provides
an easy way to express shape preservation (the shape of $G$ before
traversing the structure is the same as the shape of $G$ after
traversing it). Nevertheless, we find this notion of traversability to
be a bit unsatisfactory since it allows traversals that go over the
same element more than once.
Gibbons and Oliveira~\cite{Gibbons.Oliveira.Iterator} proposed many
properties that should hold for traversals, but they did not seek to
obtain a lawful definition of traversability. Furthermore, they failed to
recognize the law that would prevent traversing over an element twice.

The main contribution of this article is to establish coherence laws
for the distributive law that capture the intuition of traversals. The
laws can be expressed by simple equations; they are shown to hold for
the largest class known of traversable functors, namely finitary
containers; and they are shown to prohibit the ``bad'' traversals
identified by Gibbons and Oliveira. Additionally, we complete the
Haskell traversable class so that it contains the three ways in which
distributive laws can be expressed~\cite{beck:distributivity,
  barr:beck-distributivity}.  Furthermore, we characterize traversability
categorically in two ways: as a 2-functor between particular
2-categories, and as a distributive law over a monoidal action. 



The article is organised as follows. Section~\ref{sec:Haskell} is
written with the Haskell programmer in mind; we review applicative
and traversable functors, and we motivate and present the proposed
laws.
We work in the category of sets and total functions from
section~\ref{sec:traversable-functors} onwards, where we show that
finitary containers are traversable in our sense.
In section~\ref{sec:law-analysis} we analyze some consequences of the
laws and we argue that the laws imply that every position in the
structure is visited exactly once. Additionally, a categorical interpretation
of traversability is given.
Finally, in section~\ref{sec:conclusion} we
conclude and discuss future work.

\makeatletter
\@ifundefined{lhs2tex.lhs2tex.sty.read}%
  {\@namedef{lhs2tex.lhs2tex.sty.read}{}%
   \newcommand\SkipToFmtEnd{}%
   \newcommand\EndFmtInput{}%
   \long\def\SkipToFmtEnd#1\EndFmtInput{}%
  }\SkipToFmtEnd

\newcommand\ReadOnlyOnce[1]{\@ifundefined{#1}{\@namedef{#1}{}}\SkipToFmtEnd}
\usepackage{amstext}
\usepackage{amssymb}
\usepackage{stmaryrd}
\DeclareFontFamily{OT1}{cmtex}{}
\DeclareFontShape{OT1}{cmtex}{m}{n}
  {<5><6><7><8>cmtex8
   <9>cmtex9
   <10><10.95><12><14.4><17.28><20.74><24.88>cmtex10}{}
\DeclareFontShape{OT1}{cmtex}{m}{it}
  {<-> ssub * cmtt/m/it}{}
\newcommand{\texfamily}{\fontfamily{cmtex}\selectfont}
\DeclareFontShape{OT1}{cmtt}{bx}{n}
  {<5><6><7><8>cmtt8
   <9>cmbtt9
   <10><10.95><12><14.4><17.28><20.74><24.88>cmbtt10}{}
\DeclareFontShape{OT1}{cmtex}{bx}{n}
  {<-> ssub * cmtt/bx/n}{}
\newcommand{\tex}[1]{\text{\texfamily#1}}	

\newcommand{\Sp}{\hskip.33334em\relax}

\newcommand{\Conid}[1]{\mathit{#1}}
\newcommand{\Varid}[1]{\mathit{#1}}
\newcommand{\anonymous}{\kern0.06em \vbox{\hrule\@width.5em}}
\newcommand{\plus}{\mathbin{+\!\!\!+}}
\newcommand{\bind}{\mathbin{>\!\!\!>\mkern-6.7mu=}}
\newcommand{\rbind}{\mathbin{=\mkern-6.7mu<\!\!\!<}}
\newcommand{\sequ}{\mathbin{>\!\!\!>}}
\renewcommand{\leq}{\leqslant}
\renewcommand{\geq}{\geqslant}
\usepackage{polytable}

\@ifundefined{mathindent}%
  {\newdimen\mathindent\mathindent\leftmargini}%
  {}%

\def\resethooks{%
  \global\let\SaveRestoreHook\empty
  \global\let\ColumnHook\empty}
\newcommand*{\savecolumns}[1][default]%
  {\g@addto@macro\SaveRestoreHook{\savecolumns[#1]}}
\newcommand*{\restorecolumns}[1][default]%
  {\g@addto@macro\SaveRestoreHook{\restorecolumns[#1]}}
\newcommand*{\aligncolumn}[2]%
  {\g@addto@macro\ColumnHook{\column{#1}{#2}}}

\resethooks

\newcommand{\onelinecommentchars}{\quad-{}- }
\newcommand{\commentbeginchars}{\enskip\{-}
\newcommand{\commentendchars}{-\}\enskip}

\newcommand{\visiblecomments}{%
  \let\onelinecomment=\onelinecommentchars
  \let\commentbegin=\commentbeginchars
  \let\commentend=\commentendchars}

\newcommand{\invisiblecomments}{%
  \let\onelinecomment=\empty
  \let\commentbegin=\empty
  \let\commentend=\empty}

\visiblecomments

\newlength{\blanklineskip}
\setlength{\blanklineskip}{0.66084ex}

\newcommand{\hsindent}[1]{\quad}
\let\hspre\empty
\let\hspost\empty
\newcommand{\NB}{\textbf{NB}}
\newcommand{\Todo}[1]{$\langle$\textbf{To do:}~#1$\rangle$}

\EndFmtInput
\makeatother
%
%
%
%
%
%
%
%
%
\ReadOnlyOnce{polycode.fmt}%
\makeatletter

\newcommand{\hsnewpar}[1]%
  {{\parskip=0pt\parindent=0pt\par\vskip #1\noindent}}

\newcommand{\hscodestyle}{}


\newcommand{\sethscode}[1]%
  {\expandafter\let\expandafter\hscode\csname #1\endcsname
   \expandafter\let\expandafter\endhscode\csname end#1\endcsname}


\newenvironment{compathscode}%
  {\par\noindent
   \advance\leftskip\mathindent
   \hscodestyle
   \let\\=\@normalcr
   \let\hspre\(\let\hspost\)%
   \pboxed}%
  {\endpboxed\)%
   \par\noindent
   \ignorespacesafterend}

\newcommand{\compaths}{\sethscode{compathscode}}


\newenvironment{plainhscode}%
  {\hsnewpar\abovedisplayskip
   \advance\leftskip\mathindent
   \hscodestyle
   \let\hspre\(\let\hspost\)%
   \pboxed}%
  {\endpboxed%
   \hsnewpar\belowdisplayskip
   \ignorespacesafterend}

\newenvironment{oldplainhscode}%
  {\hsnewpar\abovedisplayskip
   \advance\leftskip\mathindent
   \hscodestyle
   \let\\=\@normalcr
   \(\pboxed}%
  {\endpboxed\)%
   \hsnewpar\belowdisplayskip
   \ignorespacesafterend}


\newcommand{\plainhs}{\sethscode{plainhscode}}
\newcommand{\oldplainhs}{\sethscode{oldplainhscode}}
\plainhs


\newenvironment{arrayhscode}%
  {\hsnewpar\abovedisplayskip
   \advance\leftskip\mathindent
   \hscodestyle
   \let\\=\@normalcr
   \(\parray}%
  {\endparray\)%
   \hsnewpar\belowdisplayskip
   \ignorespacesafterend}

\newcommand{\arrayhs}{\sethscode{arrayhscode}}


\newenvironment{mathhscode}%
  {\parray}{\endparray}

\newcommand{\mathhs}{\sethscode{mathhscode}}


\newenvironment{texthscode}%
  {\(\parray}{\endparray\)}

\newcommand{\texths}{\sethscode{texthscode}}


\def\codeframewidth{\arrayrulewidth}
\RequirePackage{calc}

\newenvironment{framedhscode}%
  {\parskip=\abovedisplayskip\par\noindent
   \hscodestyle
   \arrayrulewidth=\codeframewidth
   \tabular{@{}|p{\linewidth-2\arraycolsep-2\arrayrulewidth-2pt}|@{}}%
   \hline\framedhslinecorrect\\{-1.5ex}%
   \let\endoflinesave=\\
   \let\\=\@normalcr
   \(\pboxed}%
  {\endpboxed\)%
   \framedhslinecorrect\endoflinesave{.5ex}\hline
   \endtabular
   \parskip=\belowdisplayskip\par\noindent
   \ignorespacesafterend}

\newcommand{\framedhslinecorrect}[2]%
  {#1[#2]}

\newcommand{\framedhs}{\sethscode{framedhscode}}


\newenvironment{inlinehscode}%
  {\(\def\column##1##2{}%
   \let\>\undefined\let\<\undefined\let\\\undefined
   \newcommand\>[1][]{}\newcommand\<[1][]{}\newcommand\\[1][]{}%
   \def\fromto##1##2##3{##3}%
   \def\nextline{}}{\) }%

\newcommand{\inlinehs}{\sethscode{inlinehscode}}


\newenvironment{joincode}%
  {\let\orighscode=\hscode
   \let\origendhscode=\endhscode
   \def\endhscode{\def\hscode{\endgroup\def\@currenvir{hscode}\\}\begingroup}
   \orighscode\def\hscode{\endgroup\def\@currenvir{hscode}}}%
  {\origendhscode
   \global\let\hscode=\orighscode
   \global\let\endhscode=\origendhscode}%

\makeatother
\EndFmtInput

\section{Traversals in Haskell}
\label{sec:Haskell}

In this section we give our motivation for and introduce our proposed laws for
traversals using the functional language Haskell. We give an intuition
of why the laws are reasonable but we defer a more rigorous
explanation to the following sections.

Intuitively, a traversal of a data structure is a function that
collects all elements in a data structure in a given order. 
A more abstract formulation was proposed by Moggi et
al.~\cite{MoggiBJ99} where a traversal is a distributive law of a
functor (representing the data structure in question) over an
arbitrary monad. McBride and Paterson extended the notion of
traversability to be a distributive law of a functor over an arbitrary
applicative functor~\cite{mcbride08:applicative-programming}.  As we
will see next, the notion of a distributive law is general enough to
include the basic notion of a function listing the elements of a data structure.

\subsection{Applicative Functors}

An applicative functor~\cite{mcbride08:applicative-programming} is an
instance of the class

\begin{hscode}\SaveRestoreHook
\column{B}{@{}>{\hspre}l<{\hspost}@{}}%
\column{5}{@{}>{\hspre}l<{\hspost}@{}}%
\column{12}{@{}>{\hspre}l<{\hspost}@{}}%
\column{E}{@{}>{\hspre}l<{\hspost}@{}}%
\>[B]{}\mathbf{class}\;\Conid{Functor}\;\Varid{f}\Rightarrow \Conid{Applicative}\;\Varid{f}\;\mathbf{where}{}\<[E]%
\\
\>[B]{}\hsindent{5}{}\<[5]%
\>[5]{}\Varid{pure}{}\<[12]%
\>[12]{}\mathbin{::}\Varid{x}\to \Varid{f}\;\Varid{x}{}\<[E]%
\\
\>[B]{}\hsindent{5}{}\<[5]%
\>[5]{}(\oast){}\<[12]%
\>[12]{}\mathbin{::}\Varid{f}\;(\Varid{a}\to \Varid{b})\to \Varid{f}\;\Varid{a}\to \Varid{f}\;\Varid{b}{}\<[E]%
\ColumnHook
\end{hscode}\resethooks
such that the following coherence conditions hold.

\begin{hscode}\SaveRestoreHook
\column{B}{@{}>{\hspre}l<{\hspost}@{}}%
\column{17}{@{}>{\hspre}l<{\hspost}@{}}%
\column{31}{@{}>{\hspre}l<{\hspost}@{}}%
\column{60}{@{}>{\hspre}l<{\hspost}@{}}%
\column{E}{@{}>{\hspre}l<{\hspost}@{}}%
\>[B]{}\mathbf{identity}\;{}\<[17]%
\>[17]{}\qquad\qquad\;{}\<[31]%
\>[31]{}\Varid{pure}\;\Varid{id}\oast\Varid{u}{}\<[60]%
\>[60]{}\mathrel{=}\Varid{u}{}\<[E]%
\\
\>[B]{}\mathbf{composition}\;{}\<[31]%
\>[31]{}\Varid{pure}\;(\cdot)\oast\Varid{u}\oast\Varid{v}\oast\Varid{w}{}\<[60]%
\>[60]{}\mathrel{=}\Varid{u}\oast(\Varid{v}\oast\Varid{w}){}\<[E]%
\\
\>[B]{}\mathbf{homomorphism}\;{}\<[17]%
\>[17]{}\qquad\qquad\;{}\<[31]%
\>[31]{}\Varid{pure}\;\Varid{g}\oast\Varid{pure}\;\Varid{x}{}\<[60]%
\>[60]{}\mathrel{=}\Varid{pure}\;(\Varid{g}\;\Varid{x}){}\<[E]%
\\
\>[B]{}\mathbf{interchange}\;{}\<[31]%
\>[31]{}\Varid{u}\oast\Varid{pure}\;\Varid{x}{}\<[60]%
\>[60]{}\mathrel{=}\Varid{pure}\;(\lambda \Varid{g}\to \Varid{g}\;\Varid{x})\oast\Varid{u}{}\<[E]%
\ColumnHook
\end{hscode}\resethooks

Every monad is an applicative functor, but applicative functors are
more general. For example, every monoid determines an applicative
functor (which is not a monad):

\begin{hscode}\SaveRestoreHook
\column{B}{@{}>{\hspre}l<{\hspost}@{}}%
\column{6}{@{}>{\hspre}l<{\hspost}@{}}%
\column{15}{@{}>{\hspre}l<{\hspost}@{}}%
\column{E}{@{}>{\hspre}l<{\hspost}@{}}%
\>[B]{}\mathbf{newtype}\;\Conid{K}\;\Varid{a}\;\Varid{b}\mathrel{=}\Conid{K}\;\{\mskip1.5mu \Varid{unK}\mathbin{::}\Varid{a}\mskip1.5mu\}{}\<[E]%
\\[\blanklineskip]%
\>[B]{}\mathbf{instance}\;\Conid{Monoid}\;\Varid{a}\Rightarrow \Conid{Applicative}\;(\Conid{K}\;\Varid{a})\;\mathbf{where}{}\<[E]%
\\
\>[B]{}\hsindent{6}{}\<[6]%
\>[6]{}\Varid{pure}\;\Varid{x}{}\<[15]%
\>[15]{}\mathrel{=}\Conid{K}\;\emptyset{}\<[E]%
\\
\>[B]{}\hsindent{6}{}\<[6]%
\>[6]{}\Varid{f}\oast\Varid{x}{}\<[15]%
\>[15]{}\mathrel{=}\Conid{K}\;(\Varid{unK}\;\Varid{f}\oplus\Varid{unK}\;\Varid{x}){}\<[E]%
\ColumnHook
\end{hscode}\resethooks
where \ensuremath{\emptyset} is the monoid unit, and $\oplus$ is the monoid multiplication.


Applicative functors are closed under identity and
composition\footnote{For brevity, throughout the article we omit the required \ensuremath{\Conid{Functor}}
  instances. The complete source code can be found at \url{http://www.fceia.unr.edu.ar/~mauro/}. }.

\begin{hscode}\SaveRestoreHook
\column{B}{@{}>{\hspre}l<{\hspost}@{}}%
\column{5}{@{}>{\hspre}l<{\hspost}@{}}%
\column{6}{@{}>{\hspre}l<{\hspost}@{}}%
\column{16}{@{}>{\hspre}l<{\hspost}@{}}%
\column{E}{@{}>{\hspre}l<{\hspost}@{}}%
\>[B]{}\mathbf{newtype}\;\Conid{Id}\;\Varid{a}\mathrel{=}\Conid{Id}\;\{\mskip1.5mu \Varid{unId}\mathbin{::}\Varid{a}\mskip1.5mu\}{}\<[E]%
\\[\blanklineskip]%
\>[B]{}\mathbf{instance}\;\Conid{Applicative}\;\Conid{Id}\;\mathbf{where}{}\<[E]%
\\
\>[B]{}\hsindent{6}{}\<[6]%
\>[6]{}\Varid{pure}{}\<[16]%
\>[16]{}\mathrel{=}\Conid{Id}{}\<[E]%
\\
\>[B]{}\hsindent{6}{}\<[6]%
\>[6]{}\Varid{f}\oast\Varid{x}{}\<[16]%
\>[16]{}\mathrel{=}\Conid{Id}\;(\Varid{unId}\;\Varid{f}\;(\Varid{unId}\;\Varid{x})){}\<[E]%
\\[\blanklineskip]%
\>[B]{}\mathbf{newtype}\;\Conid{C}\;\Varid{f}\;\Varid{g}\;\Varid{a}\mathrel{=}\Conid{Comp}\;\{\mskip1.5mu \Varid{unC}\mathbin{::}\Varid{f}\;(\Varid{g}\;\Varid{a})\mskip1.5mu\}{}\<[E]%
\\[\blanklineskip]%
\>[B]{}\mathbf{instance}\;(\Conid{Applicative}\;\Varid{f},\Conid{Applicative}\;\Varid{g})\Rightarrow \Conid{Applicative}\;(\Conid{C}\;\Varid{f}\;\Varid{g})\;\mathbf{where}{}\<[E]%
\\
\>[B]{}\hsindent{5}{}\<[5]%
\>[5]{}\Varid{pure}{}\<[16]%
\>[16]{}\mathrel{=}\Conid{Comp}\cdot\Varid{pure}\cdot\Varid{pure}{}\<[E]%
\\
\>[B]{}\hsindent{5}{}\<[5]%
\>[5]{}\Varid{f}\oast\Varid{x}{}\<[16]%
\>[16]{}\mathrel{=}\Conid{Comp}\;(\Varid{pure}\;(\oast)\oast\Varid{unC}\;\Varid{f}\oast\Varid{unC}\;\Varid{x}){}\<[E]%
\ColumnHook
\end{hscode}\resethooks

\subsection{The class of Traversable functors}

McBride and Paterson propose a traversal to be a distributive law of a
functor over all applicative
functors~\cite{mcbride08:applicative-programming}.  Hence, the class
of traversable functors is defined:

\begin{hscode}\SaveRestoreHook
\column{B}{@{}>{\hspre}l<{\hspost}@{}}%
\column{5}{@{}>{\hspre}l<{\hspost}@{}}%
\column{14}{@{}>{\hspre}l<{\hspost}@{}}%
\column{17}{@{}>{\hspre}l<{\hspost}@{}}%
\column{E}{@{}>{\hspre}l<{\hspost}@{}}%
\>[B]{}\mathbf{class}\;\Conid{Functor}\;\Varid{t}\Rightarrow \Conid{Traversable}\;\Varid{t}\;\mathbf{where}{}\<[E]%
\\
\>[B]{}\hsindent{5}{}\<[5]%
\>[5]{}\Varid{traverse}\mathbin{::}\Conid{Applicative}\;\Varid{f}\Rightarrow (\Varid{a}\to \Varid{f}\;\Varid{b})\to \Varid{t}\;\Varid{a}\to \Varid{f}\;(\Varid{t}\;\Varid{b}){}\<[E]%
\\
\>[B]{}\hsindent{5}{}\<[5]%
\>[5]{}\Varid{dist}{}\<[14]%
\>[14]{}\mathbin{::}\Conid{Applicative}\;\Varid{f}\Rightarrow \Varid{t}\;(\Varid{f}\;\Varid{a})\to \Varid{f}\;(\Varid{t}\;\Varid{a}){}\<[E]%
\\[\blanklineskip]%
\>[B]{}\hsindent{5}{}\<[5]%
\>[5]{}\Varid{traverse}\;\Varid{f}{}\<[17]%
\>[17]{}\mathrel{=}\Varid{dist}\cdot\Varid{fmap}\;\Varid{f}{}\<[E]%
\\
\>[B]{}\hsindent{5}{}\<[5]%
\>[5]{}\Varid{dist}{}\<[17]%
\>[17]{}\mathrel{=}\Varid{traverse}\;\Varid{id}{}\<[E]%
\ColumnHook
\end{hscode}\resethooks

A minimal instance should provide a definition of either \ensuremath{\Varid{traverse}} or
\ensuremath{\Varid{dist}}, as one can be defined by the other, as shown by the default
instances above.

\begin{example} The canonical example of a traversable functor is the
  list functor.
\begin{hscode}\SaveRestoreHook
\column{B}{@{}>{\hspre}l<{\hspost}@{}}%
\column{5}{@{}>{\hspre}l<{\hspost}@{}}%
\column{19}{@{}>{\hspre}l<{\hspost}@{}}%
\column{E}{@{}>{\hspre}l<{\hspost}@{}}%
\>[B]{}\mathbf{instance}\;\Conid{Traversable}\;[\mskip1.5mu \mskip1.5mu]\;\mathbf{where}{}\<[E]%
\\
\>[B]{}\hsindent{5}{}\<[5]%
\>[5]{}\Varid{dist}\;[\mskip1.5mu \mskip1.5mu]{}\<[19]%
\>[19]{}\mathrel{=}\Varid{pure}\;[\mskip1.5mu \mskip1.5mu]{}\<[E]%
\\
\>[B]{}\hsindent{5}{}\<[5]%
\>[5]{}\Varid{dist}\;(\Varid{x}\mathbin{:}\Varid{xs}){}\<[19]%
\>[19]{}\mathrel{=}\Varid{pure}\;(\mathbin{:})\oast\Varid{x}\oast\Varid{dist}\;\Varid{xs}{}\<[E]%
\ColumnHook
\end{hscode}\resethooks
 
Another typical example is that of binary trees with information in
the nodes:
\begin{hscode}\SaveRestoreHook
\column{B}{@{}>{\hspre}l<{\hspost}@{}}%
\column{5}{@{}>{\hspre}l<{\hspost}@{}}%
\column{24}{@{}>{\hspre}l<{\hspost}@{}}%
\column{E}{@{}>{\hspre}l<{\hspost}@{}}%
\>[B]{}\mathbf{data}\;\Conid{Bin}\;\Varid{a}\mathrel{=}\Conid{Leaf}\mid \Conid{Node}\;(\Conid{Bin}\;\Varid{a})\;\Varid{a}\;(\Conid{Bin}\;\Varid{a}){}\<[E]%
\\[\blanklineskip]%
\>[B]{}\mathbf{instance}\;\Conid{Traversable}\;\Conid{Bin}\;\mathbf{where}{}\<[E]%
\\
\>[B]{}\hsindent{5}{}\<[5]%
\>[5]{}\Varid{dist}\;\Conid{Leaf}{}\<[24]%
\>[24]{}\mathrel{=}\Varid{pure}\;\Conid{Leaf}{}\<[E]%
\\
\>[B]{}\hsindent{5}{}\<[5]%
\>[5]{}\Varid{dist}\;(\Conid{Node}\;\Varid{l}\;\Varid{a}\;\Varid{r}){}\<[24]%
\>[24]{}\mathrel{=}\Varid{pure}\;\Conid{Node}\oast\Varid{dist}\;\Varid{l}\oast\Varid{a}\oast\Varid{dist}\;\Varid{r}{}\<[E]%
\ColumnHook
\end{hscode}\resethooks

Our last example is the identity functor.
\begin{hscode}\SaveRestoreHook
\column{B}{@{}>{\hspre}l<{\hspost}@{}}%
\column{5}{@{}>{\hspre}l<{\hspost}@{}}%
\column{18}{@{}>{\hspre}l<{\hspost}@{}}%
\column{E}{@{}>{\hspre}l<{\hspost}@{}}%
\>[B]{}\mathbf{instance}\;\Conid{Traversable}\;\Conid{Id}\;\mathbf{where}{}\<[E]%
\\
\>[B]{}\hsindent{5}{}\<[5]%
\>[5]{}\Varid{dist}\;(\Conid{Id}\;\Varid{x}){}\<[18]%
\>[18]{}\mathrel{=}\Varid{fmap}\;\Conid{Id}\;\Varid{x}{}\<[E]%
\ColumnHook
\end{hscode}\resethooks

\end{example}


\begin{remark}
  Given a \ensuremath{\Conid{Traversable}} functor, it is possible to construct a list of
  its elements by traversing it over an accumulator, as done by the
  function \ensuremath{\Varid{toList}}:

\begin{hscode}\SaveRestoreHook
\column{B}{@{}>{\hspre}l<{\hspost}@{}}%
\column{6}{@{}>{\hspre}l<{\hspost}@{}}%
\column{9}{@{}>{\hspre}l<{\hspost}@{}}%
\column{13}{@{}>{\hspre}l<{\hspost}@{}}%
\column{21}{@{}>{\hspre}l<{\hspost}@{}}%
\column{E}{@{}>{\hspre}l<{\hspost}@{}}%
\>[B]{}\Varid{toList}{}\<[9]%
\>[9]{}\mathbin{::}\Conid{Traversable}\;\Varid{t}\Rightarrow \Varid{t}\;\Varid{a}\to [\mskip1.5mu \Varid{a}\mskip1.5mu]{}\<[E]%
\\
\>[B]{}\Varid{toList}{}\<[9]%
\>[9]{}\mathrel{=}\Varid{unK}\cdot\Varid{dist}\cdot\Varid{fmap}\;\Varid{wrap}{}\<[E]%
\\
\>[B]{}\hsindent{6}{}\<[6]%
\>[6]{}\mathbf{where}\;{}\<[13]%
\>[13]{}\Varid{wrap}{}\<[21]%
\>[21]{}\mathbin{::}\Varid{x}\to \Conid{K}\;[\mskip1.5mu \Varid{x}\mskip1.5mu]\;\Varid{a}{}\<[E]%
\\
\>[13]{}\Varid{wrap}\;\Varid{x}{}\<[21]%
\>[21]{}\mathrel{=}\Conid{K}\;[\mskip1.5mu \Varid{x}\mskip1.5mu]{}\<[E]%
\ColumnHook
\end{hscode}\resethooks
\end{remark}

\subsection{The need for laws}
\label{sec:need_for_laws}
Notice that no coherence conditions are required to hold for instances
of \ensuremath{\Varid{traverse}}
(or \ensuremath{\Varid{dist}}) in the class definition above.  Hence, it is possible to
instantiate definitions of \ensuremath{\Varid{dist}} that do not follow our intuition of
what a traversal should be.
%
Consider the following functions implementing a distributive law
between lists and an arbitrary applicative functor.

\begin{hscode}\SaveRestoreHook
\column{B}{@{}>{\hspre}l<{\hspost}@{}}%
\column{24}{@{}>{\hspre}l<{\hspost}@{}}%
\column{28}{@{}>{\hspre}l<{\hspost}@{}}%
\column{E}{@{}>{\hspre}l<{\hspost}@{}}%
\>[B]{}\Varid{distL},\Varid{distL'},\Varid{distL''}{}\<[24]%
\>[24]{}\mathbin{::}\Conid{Applicative}\;\Varid{f}\Rightarrow [\mskip1.5mu \Varid{f}\;\Varid{a}\mskip1.5mu]\to \Varid{f}\;[\mskip1.5mu \Varid{a}\mskip1.5mu]{}\<[E]%
\\
\>[B]{}\Varid{distL}\;\anonymous {}\<[24]%
\>[24]{}\mathrel{=}\Varid{pure}\;[\mskip1.5mu \mskip1.5mu]{}\<[E]%
\\[\blanklineskip]%
\>[B]{}\Varid{distL'}\;[\mskip1.5mu \mskip1.5mu]{}\<[24]%
\>[24]{}\mathrel{=}\Varid{pure}\;[\mskip1.5mu \mskip1.5mu]{}\<[E]%
\\
\>[B]{}\Varid{distL'}\;[\mskip1.5mu \Varid{x}\mskip1.5mu]{}\<[24]%
\>[24]{}\mathrel{=}\Varid{pure}\;[\mskip1.5mu \mskip1.5mu]{}\<[E]%
\\
\>[B]{}\Varid{distL'}\;(\Varid{x}\mathbin{:}\Varid{y}\mathbin{:}\Varid{xs}){}\<[24]%
\>[24]{}\mathrel{=}\Varid{pure}\;(\mathbin{:})\oast\Varid{x}\oast\Varid{distL'}\;(\Varid{y}\mathbin{:}\Varid{xs}){}\<[E]%
\\[\blanklineskip]%
\>[B]{}\Varid{distL''}\;[\mskip1.5mu \mskip1.5mu]{}\<[24]%
\>[24]{}\mathrel{=}\Varid{pure}\;[\mskip1.5mu \mskip1.5mu]{}\<[E]%
\\
\>[B]{}\Varid{distL''}\;(\Varid{x}\mathbin{:}\Varid{xs}){}\<[24]%
\>[24]{}\mathrel{=}\Varid{pure}\;(\mathbin{:})\oast\Varid{x'}\oast\Varid{distL''}\;\Varid{xs}{}\<[E]%
\\
\>[24]{}\hsindent{4}{}\<[28]%
\>[28]{}\mathbf{where}\;\Varid{x'}\mathrel{=}\Varid{pure}\;(\lambda \Varid{x}\;\Varid{y}\to \Varid{x})\oast\Varid{x}\oast\Varid{x}{}\<[E]%
\ColumnHook
\end{hscode}\resethooks
The functions have the correct type but they are not traversals: the
first one (\ensuremath{\Varid{distL}}) does
not even try to traverse the list, the second one (\ensuremath{\Varid{distL'}}) does not
visit the last element of the list, and the last one (\ensuremath{\Varid{distL''}}) collects
each applicative effect twice (but not the data).
Consequently, in order for the \ensuremath{\Conid{Traversable}} class to capture the
intuition behind traversals, it needs to require certain laws to hold
in order to prohibit definitions such as those of the three functions
above.

\subsection{Reformulation of the \ensuremath{\Conid{Traversable}} class}

We conclude this section with our proposed class of Traversable functors. 
It differs from the previous one in two aspects.
\begin{itemize}
\item In addition to \ensuremath{\Varid{traverse}} and \ensuremath{\Varid{dist}}, a minimal instance can
  also be given by defining a function
 \ensuremath{\Varid{consume}\mathbin{::}\Conid{Applicative}\;\Varid{f}\Rightarrow (\Varid{t}\;\Varid{a}\to \Varid{b})\to \Varid{t}\;(\Varid{f}\;\Varid{a})\to \Varid{f}\;\Varid{b}}.
\item It requires two laws to hold. 
\end{itemize}

We define the class of traversable functors to be:

\begin{hscode}\SaveRestoreHook
\column{B}{@{}>{\hspre}l<{\hspost}@{}}%
\column{5}{@{}>{\hspre}l<{\hspost}@{}}%
\column{15}{@{}>{\hspre}l<{\hspost}@{}}%
\column{18}{@{}>{\hspre}l<{\hspost}@{}}%
\column{E}{@{}>{\hspre}l<{\hspost}@{}}%
\>[B]{}\mathbf{class}\;\Conid{Functor}\;\Varid{t}\Rightarrow \Conid{Traversable}\;\Varid{t}\;\mathbf{where}{}\<[E]%
\\
\>[B]{}\hsindent{5}{}\<[5]%
\>[5]{}\Varid{traverse}\mathbin{::}\Conid{Applicative}\;\Varid{f}\Rightarrow (\Varid{a}\to \Varid{f}\;\Varid{b})\to \Varid{t}\;\Varid{a}\to \Varid{f}\;(\Varid{t}\;\Varid{b}){}\<[E]%
\\
\>[B]{}\hsindent{5}{}\<[5]%
\>[5]{}\Varid{dist}{}\<[15]%
\>[15]{}\mathbin{::}\Conid{Applicative}\;\Varid{f}\Rightarrow \Varid{t}\;(\Varid{f}\;\Varid{a})\to \Varid{f}\;(\Varid{t}\;\Varid{a}){}\<[E]%
\\
\>[B]{}\hsindent{5}{}\<[5]%
\>[5]{}\Varid{consume}{}\<[15]%
\>[15]{}\mathbin{::}\Conid{Applicative}\;\Varid{f}\Rightarrow (\Varid{t}\;\Varid{a}\to \Varid{b})\to \Varid{t}\;(\Varid{f}\;\Varid{a})\to \Varid{f}\;\Varid{b}{}\<[E]%
\\[\blanklineskip]%
\>[B]{}\hsindent{5}{}\<[5]%
\>[5]{}\Varid{traverse}\;\Varid{f}{}\<[18]%
\>[18]{}\mathrel{=}\Varid{dist}\cdot\Varid{fmap}\;\Varid{f}{}\<[E]%
\\
\>[B]{}\hsindent{5}{}\<[5]%
\>[5]{}\Varid{dist}{}\<[18]%
\>[18]{}\mathrel{=}\Varid{consume}\;\Varid{id}{}\<[E]%
\\
\>[B]{}\hsindent{5}{}\<[5]%
\>[5]{}\Varid{consume}\;\Varid{f}{}\<[18]%
\>[18]{}\mathrel{=}\Varid{fmap}\;\Varid{f}\cdot\Varid{traverse}\;\Varid{id}{}\<[E]%
\ColumnHook
\end{hscode}\resethooks

subject to the following laws:

\begin{hscode}\SaveRestoreHook
\column{B}{@{}>{\hspre}l<{\hspost}@{}}%
\column{5}{@{}>{\hspre}l<{\hspost}@{}}%
\column{20}{@{}>{\hspre}l<{\hspost}@{}}%
\column{32}{@{}>{\hspre}l<{\hspost}@{}}%
\column{51}{@{}>{\hspre}l<{\hspost}@{}}%
\column{E}{@{}>{\hspre}l<{\hspost}@{}}%
\>[5]{}\mathbf{Unitarity}\;{}\<[20]%
\>[20]{}\qquad\qquad\;{}\<[32]%
\>[32]{}\Varid{dist}\cdot\Varid{fmap}\;\Conid{Id}{}\<[51]%
\>[51]{}\mathrel{=}\Conid{Id}{}\<[E]%
\\
\>[5]{}\mathbf{Linearity}\;{}\<[20]%
\>[20]{}\qquad\qquad\;{}\<[32]%
\>[32]{}\Varid{dist}\cdot\Varid{fmap}\;\Conid{Comp}{}\<[51]%
\>[51]{}\mathrel{=}\Conid{Comp}\cdot\Varid{fmap}\;\Varid{dist}\cdot\Varid{dist}{}\<[E]%
\ColumnHook
\end{hscode}\resethooks

The same laws expressed in terms of \ensuremath{\Varid{traverse}} are:
\begin{hscode}\SaveRestoreHook
\column{B}{@{}>{\hspre}l<{\hspost}@{}}%
\column{5}{@{}>{\hspre}l<{\hspost}@{}}%
\column{35}{@{}>{\hspre}l<{\hspost}@{}}%
\column{E}{@{}>{\hspre}l<{\hspost}@{}}%
\>[5]{}\Varid{traverse}\;(\Conid{Id}\cdot\Varid{f}){}\<[35]%
\>[35]{}\mathrel{=}\Conid{Id}\cdot\Varid{fmap}\;\Varid{f}{}\<[E]%
\\
\>[5]{}\Varid{traverse}\;(\Conid{Comp}\cdot\Varid{fmap}\;\Varid{g}\cdot\Varid{f}){}\<[35]%
\>[35]{}\mathrel{=}\Conid{Comp}\cdot\Varid{fmap}\;(\Varid{traverse}\;\Varid{g})\cdot\Varid{traverse}\;\Varid{f}{}\<[E]%
\ColumnHook
\end{hscode}\resethooks

We leave as an exercise to the reader to express the laws in terms of \ensuremath{\Varid{consume}}.

Note that, since the identity functor is applicative, the laws for
traverse imply functoriality even without the type class
requirement. Hence an alternative definition would not require
traversable functors to be an instance of the \ensuremath{\Conid{Functor}} class, but
just to provide an instance of \ensuremath{\Varid{traverse}} subject to the two laws
above. In this case, providing definitions just for \ensuremath{\Varid{dist}} or
\ensuremath{\Varid{consume}} would not be enough.

In the following sections we show that these laws are reasonable for a
large class of functors, and
analyse how they prohibit wrong definitions such as \ensuremath{\Varid{distL}}, \ensuremath{\Varid{distL'}}
and \ensuremath{\Varid{distL''}}.



\section{Canonical Traverse for Finitary Containers}
\label{sec:traversable-functors}
Categorically, \emph{applicative functors} are functors
$\fF: \Set \larr \Set$, together with natural transformations:
\begin{align*}
\eta_X &~:~ X \larr \fF X&\text{(unit)}\\
\circledast_{X,Y} &~:~ \fF(\fun{X}{Y})\x \fF X \larr
  \fF Y&\text{(application under $\fF$)}&\enspace,
\end{align*}
where $\fun{X}{Y}$ is the function space (cartesian closure),
together with equations expressing basically that $\eta_{\fun{X}{Y}}$
injects pure functions $\fun{X}{Y}$ to functions under in
$\fF(\fun{X}{Y})$ and $\circledast$ respects this (see \cite{mcbride08:applicative-programming} for the
details). 

In category theory it is more convenient to work with the following equivalent definitions:
\begin{definition}[Applicative Functor]\label{def:appl-func}An \emph{applicative functor}
  is equivalently either of:
\begin{enumerate}[(i)]
\item\label{def:appl-func:it1} A functor $\fF: \Set \larr \Set$ which is lax monoidal with
  respect to the cartesian product and whose strength is coherent with
  the monoidal structure, which is to say that the following diagram
  commutes:
\[
\xymatrix@C+=2cm{(\fF X \x \fF Y) \x Z \ar[r]^{\alpha}\ar[d]_{\mu\x Z}&\fF X \x (\fF Y
  \x Z)\ar[r]^{\fF X \x \sigma} & \fF X \x \fF(Y \x Z)\ar[d]^{\mu}\\
  \fF(X \x Y) \x Z \ar[r]^{\sigma} &\fF((X \x Y) \x Z) \ar[r]^{\fF\alpha} & \fF(X \x (Y\x Z))}
\]
where $\sigma$ is strength, $\mu$ is the monoidal action and $\alpha$
is associativity. Note that all $\Set$ functors are strong so the key
requirement here is the coherence with the monoidal action. 

\item\label{def:appl-func:it2} A pointed lax monoidal functor $\fF : \Set \larr \Set$, where
  the unit of $\fF$, $\eta_X: X \larr \fF X$, coincides with the unit
  of the monoidal structure $\nu : 1 \larr \fF 1$ in that $\eta_1 =
  \nu$; and the multiplication $\mu_{X,Y} : \fF X \x \fF Y \larr \fF(X
  \x Y)$ is coherent with $\eta$ in the sense that 
\[\xymatrix{
 & X \x Y \ar[dl]_{\eta_X \x \eta_Y} \ar[dr]^{\eta_{X\x Y}}
\\
\fF X \x \fF Y \ar[rr]_{\mu_{X,Y}} & & \fF(X \x Y)}\]
commutes. 
\end{enumerate}
\end{definition}
\noindent
The equivalence of \eqref{def:appl-func:it1} and
\eqref{def:appl-func:it2} is straightforward.
%

\begin{definition}[Applicative Morphism]
  Let $F$ and $G$ be applicative functors. An \emph{applicative
    morphism} is a natural transformation $\tau : F \to G$ that
  respects the unit and multiplication. That is, a natural
  transformation $\tau$ such that the following diagrams commute.
\[
\xymatrix{
& X \ar[dl]_{\eta^F_X} \ar[dr]^{\eta^G_X}& \\
FX \ar[rr]_{\tau_X} & & GX
}
\qquad
\xymatrix@C+=2cm{
FX\times FY
    \ar[r]^{\mu^F_{X,Y}}
    \ar[d]_{\tau_X\times\tau_Y}
 &
F(X\times Y)
    \ar[d]^{\tau_{X\times Y}}
\\
GX\times GY 
  \ar[r]_{\mu^G_{X,Y}}
&
G(X\times Y)
}
\]  
\end{definition}

Applicative functors and applicative morphisms form a
category $\cA$.

The identity functor $\Id$ is an applicative functor, and composition
of applicative functors is applicative. Hence, applicative functors
form a (large) monoid, with functor composition $\circ$ as
multiplication and identity functor $\Id$ as unit.

In \cite{mcbride08:applicative-programming}, \emph{traversable
  functors} are characterised as those which \emph{distribute over}
all applicative functors. There, distributivity of a traversable $\fT$
over applicative $\fF$'s meant only the existence of natural
transformations of type $\fT \fF \Rightarrow \fF \fT $. However, without
further constraints this characterisation is too coarse. Here we
refine the notion of a traversable functor as follows: 
\begin{definition}[Traversable Functor]
A functor $\fT : \Set \larr \Set$  is said to be \emph{traversable}
if there is a family of natural transformations
\[\delta^\fF_X : \fT \fF X \larr \fF \fT X\]
natural in $\fF$ and respecting the monoidal structure of applicative
functor composition. Explicitly, for all applicative $\fF$, $\fG :
\Set \larr \Set$ and applicative morphisms $\alpha : \fF
\larr \fG$, the following diagrams of natural transformations commute:
\[
\begin{array}{ccc}
\xymatrix{ \fT\fF \ar[r]^{\delta^\fF}\ar[d]_{\fT\alpha} &\fF
  \fT \ar[d]^{\alpha_{\fT}}\\ \fT\fG \ar[r]_{\delta^{\fG}}&
  \fG\fT}
&  
\xymatrix{&\fF\fT\fG  \ar[dr]^{\fF\delta^{\fG}} \\\fT\fF\fG\ar[rr]_{\delta^{\fF\fG}} \ar[ur]^{\delta^F\fG}& & \fF\fG\fT}
&
\xymatrix{ \fT\Id \ar@/_1em/[r]_{\delta^\Id}\ar@/^1em/[r]^{\iid_\fT} &
  \Id\fT }\\
\text{naturality}& \text{linearity}&\text{unitarity}
\end{array}
\]
We sometimes call the family $\delta$ a \emph{traversal of $\fT$}.
\end{definition}
%


%
Next we introduce a class of functors, which are always traversable:
so-called\/ \emph{finitary containers}~\cite{alti:fossacs03}. 
\begin{definition}[Finitary Container]
 A \emph{finitary container} is given by 
 \begin{enumerate}[(i)]
 \item a set $S$ of \emph{shapes}
 \item an \emph{arity} $\mathsf{ar} : S \to \Nat$
 \end{enumerate}
To each container $(S,\mathsf{ar})$ one can assign a functor $\Ext(S,\mathsf{ar}) : \Set
\larr \Set$ called the
\emph{extension of $(S,\mathsf{ar})$} defined for each set $X$ as the set of (dependent) pairs
$(s,f)$ where $s \in S$ and $f \in X ^ {\mathsf{ar}\,s}$, where
$X^n$, for $n\in \Nat$, is the $n$-fold product $\underbrace{X \x \cdots
  \x X}_{n~\text{times}}$
\end{definition}
Finitary containers are also known as finitary \emph{dependent
  polynomial functors}~\cite{gambino04:wellfounded-trees} or functors
\emph{shapely over lists}~\cite{MoggiBJ99}. Moggi et
al.~\cite{MoggiBJ99} define a canonical traversal by monads for
shapely functors.  It is also indicated that this traversal could be
generalised from monads to all monoidal functors. Here we show that
all finitary containers (shapely functors) are traversable in our
sense. Explicitly: for each extension of a finitary container we
construct a natural family of distributive laws $\delta$ satisfying
\emph{linearity} and \emph{unitarity}.

As extensions of containers are
sums of products we proceed in stages. 
\begin{lemma}\label{lem:sums-are-traversable}
  Traversable functors are closed under arbitrary sums.
\end{lemma}
\begin{proof}
Let $\fT_s$, $s \in S$ be a family of traversable functors with
traversals $\delta_s$. We must show that $\sum_{s \in S}\fT_s$ is
traversable.  To this end we construct, for each $X$:
\[
  \xymatrix{
    \sum_{s\in S}\fT_s \fF X \ar[rr]^{\sum_{s\in S}\delta_s} &&
    \sum_{s\in S}\fF\fT_s X \ar[rr]^{[\,\fF(\mathsf{inj}_s\,\fT_s X)]_{s\in S}} && F\sum_{s\in S}\fT_s X}\enspace,
\]
where $\mathsf{inj}_s y = (s , y)$ and where $[~]_{s \in S}$ is
case splitting on $S$. Naturality, unitarity and linearity are all
easily checked. 
\end{proof}
\begin{lemma}\label{lem:prods-are-traversable}
   Traversable functors are closed under finite products. 
\end{lemma}
\begin{proof}
  Let $\delta$ be a traversal for $\fT$, $S$
  finite. Then $\fT^S$ is just iterated product $\fT \x \cdots \x \fT$
  , $|S|$-times, and therefore we can use multiplication $\mu$
  of any applicative functor $\fF$ to construct a $\delta^S : \fT^S
  \fF \Rightarrow \fF \fT^S$ by finite iteration. Namely, define
  $\mu^k : F^k X \to F (X^k)$, $k\in \Nat$ by $\mu^0 = \nu$, $\mu^1 =
  \Id_F$, $\mu^{k+1} = \mu \cdot (\Id_F \x \mu^k)$ and put $\delta^S =
  \mu^{|S|} \cdot \delta^{|S|}$. 
  Now naturality
  follows by naturality of everything in sight; unitarity from the fact
  that $\mu$ for the identity monoidal functor is just the
  identity. To see linearity just observe commutativity of the
  following for each $k\in \Nat$:
\[
\xymatrix{ & \fF(\fG X)^k \ar[dr]^{\fF(\mu^G)^k_X}
  \\
  (\fF\fG X)^k\ar[rr]_{(\fF\mu^G\cdot\mu^F)^k_X}
  \ar[ur]^{(\mu^F)^k_{GX}}& & \fF\fG (X^k)} \enspace.\] 
\end{proof}
\begin{theorem}\label{thm:finitary-traversable}
  All extensions of finitary containers are traversable.
\end{theorem}
\begin{proof}
  Extensions of finitary containers are sums of finite products so the
  results follows directly by the previous two Lemmas
  \ref{lem:prods-are-traversable} and \ref{lem:sums-are-traversable}.
\end{proof}
It remains an open question whether the traversals defined in the
above theorem are essentially unique. In other words, whether there is
an isomorphism between permutations of arities on finitary containers
and their traversals.


\newcommand{\pure}{\eta}
\section{Analysis of the Laws}
\label{sec:law-analysis}

\subsection{Unitarity}\label{sec:unitarity}
Unitarity implies the so-called ``purity law''
$\delta^F \cdot T(\pure) =
\pure_{T}$~\cite{Gibbons.Oliveira.Iterator}. In fact, the two laws are equivalent.

To see this, we notice that the identity functor is initial in the
category $\cA$ of
applicative functors and applicative morphisms, with the
universal map given by $\pure$. Hence we obtain the purity law by the
following naturality square.

\[
\xymatrix@C+=2cm{
T \Id \ar[r]^{\delta^\Id=\iid} 
      \ar[d]_{T(\pure)}
& \Id T \ar[d]^{\pure_T}
\\
TF  \ar[r]_{\delta^F}
&
FT
}
\]

Conversely, instantiating the purity law for the identity functor, we
obtain unitarity.

Unitarity (and hence, the purity law) is stronger
than shape preservation. For example, returning the mirror of a binary
tree is forbidden. 
By requiring unitarity to hold we are
implicitly stating that what matters in a traversal is the order in
which the effects are collected.

Note that unitarity implies that each element is
traversed at least once. Hence distributive laws such as $distL$ and
$distL'$ of Section~\ref{sec:need_for_laws} do not respect unitarity.

\subsection{Linearity}

One of the sought-after effects of the laws is to rule out traversals
that go more than once over each element.  We will show that this
kind of definition does not respect the linearity
law by analysing the simple example of traversing the identity
functor.

Consider the following distributive law of the identity functor over
an arbitrary applicative functor,
\[
\delta^{F}_{X} : 
\xymatrix@C+=2cm{ \Id(FX)=FX \ar[r]^{d_{FX}} & 
                  FX\times FX \ar[r]^{\mu^F_X} & 
                  F(X\times X) \ar[r]^{F\pi_1} & FX=F(\Id\,X)}
\]
where $d_X (x : X) = (x,x) : X\times X$ is the diagonal function. The
distributive law $\delta$ traverses over the data twice so it is not a
proper traversal. Although unitarity holds for $\delta$, linearity
does not, as the following 
counter-example shows.

Consider the applicative functor $L$ arising from the list monad,
and $[[],[1]] \in L(L(\Nat))$. Calculating, we obtain that linearity does not hold:
\[
\delta^{L\circ L}_\Nat [[],[1]] \quad=\quad [[],[],[],[1]] \quad\not=\quad [[],[1],[],[1]]
\quad=\quad \delta^L_\Nat (L \delta^L_\Nat [[],[1]]) 
\]

A similar counter-example can be constructed for the traversal of
lists $distL''$ of Section~\ref{sec:need_for_laws}.
\[
distL''^{L\circ L} [[[], [[1]]]] \ =\  [[],[],[],[[[1]]]] \quad\not=\quad
[[],[],[[[1]]],[[[1]]]] \ =\  distL''^L \circ L (distL''^L) [[[], [[1]]]]
\]





These counter-examples suggest that distributive laws that traverse
over elements more than once will violate the linearity law. This 
consequence of the linearity law was apparently overlooked in
\cite{Gibbons.Oliveira.Iterator}. 

\subsection{Preservation of Kleisli Composition}

The following lemma was proved as a property of the canonical
distributivity of finitary containers by Moggi et al.~\cite{MoggiBJ99}
and under label \emph{sequential composition of monadic traversals} in
\cite{Gibbons.Oliveira.Iterator}. Here, however, we can prove it as a
consequence of the definition of traversability.

\begin{lemma}
Let $(T,\delta)$ be a traversable functor. Then for any commutative
monad $M$, we have that $\delta$ preserves Kleisli composition, i.e.
\[
\xymatrix@C+=2cm{
TMM \ar[r]^{\delta^{M}M}
   \ar[d]_{T\flatten}
  & MTM \ar[r]^{M\delta^M}& MMT \ar[d]^{\flatten T}
\\
TM \ar[rr]_{\delta^M}& & MT
}
\]
where $\flatten$ is the multiplication of the monad $M$.
\end{lemma}
\begin{proof}
  The commutative monad $M$ induces a commutative applicative functor.
  The multiplication of the monad $\flatten : M\circ M \to M$ is an
  applicative morphism (thanks to commutativity). Hence,
  the following diagram commutes because of preservation of
  composition and naturality of $\delta$.
\[
\xymatrix@C+=2cm{
TMM \ar[r]^{\delta^{M}M}
   \ar[d]_{T\flatten}
   \ar@/_2em/[rr]_{\delta^{M\circ M}}
  & MTM \ar[r]^{M\delta^M}& MMT \ar[d]^{\flatten T}
\\
TM \ar[rr]_{\delta^M}& & MT
}
\]
\end{proof}

 

\subsection{Categorical meaning of the laws}
%
Let \textsf{App} be the 2-category with one object, arrows:
applicative functors and 2-cells: applicative morphisms. And let $J:
\mathsf{App} \larr \mathsf{Cat}$ be the inclusion into the
2-category of categories sending the object of $\mathsf{App}$ to
$\Set$. Then a traversable functor is exactly a co-lax natural
transformation $J \larr J : \mathsf{App} \larr \mathsf{Cat}$. See
\cite{kelly74:elements-of-2categories} for the elementary 2-categorical notions.

It's interesting to compare this to the following characterisation in
terms of distributive laws of monoidal actions
\cite{skoda04:distributive-laws}.  Consider the strict monoidal
category $(\mathcal{A},\circ,\Id)$ of applicative functors and
applicative morphisms with the monoidal action
applicative functor composition. This category induces a monoidal
action over $\Set$, where the action
$\lozenge:\mathcal{A}\times\Set\to\Set$ is simply
application of the functor, i.e. $\fF \lozenge X \mapsto \fF X$ and
$\fF \lozenge f \mapsto \fF f$.

A distributive law of a functor $T$ over the monoidal action $\lozenge$
of $\cA$ is a binatural transformation $\delta:T(\mathunderscore\lozenge\mathunderscore)
\to \mathunderscore\lozenge(T\mathunderscore)$, satisfying the axioms:

\[
\xymatrix{
T((F\circ G) \lozenge X) 
  \ar[d]_{\iid} \ar[rr]^{\delta^{F\circ G}_X}
& & (F\circ G) \lozenge T X  \ar[d]^{\iid}
\\
T(F \lozenge (G \lozenge X)) \ar[r]_{\delta^F_{G\lozenge X}}
 & F \lozenge T(G\lozenge X) \ar[r]_{F \lozenge \delta^G_X}
 & F \lozenge (G \lozenge TX)
}
\qquad
\xymatrix{
 & TX \ar[dl]_{T\iid} 
      \ar[dr]^{\iid}
\\
T(\Id \lozenge X) 
\ar[rr]_{\delta^{\Id}_X}
& & \Id \lozenge T X
}
\]

Hence, a functor is traversable when it comes equipped with a
distributive law over the monoidal action of applicative functors.




\section{Conclusion}
\label{sec:conclusion}

The definition of traversability as a distributive law needs to be
strenghtened in order to capture the intuitive notion of a traversal.
We have provided two simple laws and shown that these laws hold for finitary
containers, and we have provided evidence that they restrict
traversals to those that go over each position exactly once.


Finitary containers are the largest known class of traversable
functors, and we are not aware of any functor that is traversable and it
is not a finitary container. It has been conjectured that
every traversable functor is a finitary container~\cite{MoggiBJ99}, but
proof of this has eluded us. One possiblity for proving it might be to
take a syntactic approach and restrict the proof to functors in a
given universe.

The characterisation of traversals as a distributive law over
applicative functors is interesting because it leads to other kinds of
traversability by changing distributivity over applicative functors by
distributivity over others kinds of functor.  For example, one might
consider a distributive law over all commutative applicative functors
 i.e. applicative functors $F$ for which
\newcommand{\swap}{\mathsf{swap}} 
$$F\swap_{X,Y}\circ\mu = \mu \circ \swap_{FX,FY} :
FX\times FY \to F(Y\times X)\enspace,$$
where $\swap_{X,Y} : X\times Y \to Y\times X$ is the obvious morphism.
Then, one would obtain a class of traversable functors that includes finitary
quotient containers~\cite{alti:mpc04} such as unordered pairs, where
commutativity of the applicative functor is essential.

\paragraph{Acknowledgement}

We thank the anonymous reviewers for their constructive criticism.





\bibliographystyle{eptcs}
\bibliography{bib.bib}

\begin{thebibliography}{10}
\providecommand{\bibitemdeclare}[2]{}
\providecommand{\urlprefix}{Available at }
\providecommand{\url}[1]{\texttt{#1}}
\providecommand{\href}[2]{\texttt{#2}}
\providecommand{\urlalt}[2]{\href{#1}{#2}}
\providecommand{\doi}[1]{doi:\urlalt{http://dx.doi.org/#1}{#1}}
\providecommand{\bibinfo}[2]{#2}

\bibitemdeclare{inproceedings}{alti:fossacs03}
\bibitem{alti:fossacs03}
\bibinfo{author}{Michael Abbott}, \bibinfo{author}{Thorsten Altenkirch} \&
  \bibinfo{author}{Neil Ghani} (\bibinfo{year}{2003}):
  \emph{\bibinfo{title}{Categories of Containers}}.
\newblock In: {\sl \bibinfo{booktitle}{Proceedings of Foundations of Software
  Science and Computation Structures}}, pp. \bibinfo{pages}{23--38},
  \doi{10.1007/3-540-36576-1\_2}.

\bibitemdeclare{inproceedings}{alti:mpc04}
\bibitem{alti:mpc04}
\bibinfo{author}{Michael Abbott}, \bibinfo{author}{Thorsten Altenkirch},
  \bibinfo{author}{Neil Ghani} \& \bibinfo{author}{Conor McBride}
  (\bibinfo{year}{2004}): \emph{\bibinfo{title}{Constructing Polymorphic
  Programs with Quotient Types}}.
\newblock In: {\sl \bibinfo{booktitle}{7th International Conference on
  Mathematics of Program Construction (MPC 2004)}}, pp. \bibinfo{pages}{2--15},
  \doi{10.1007/978-3-540-27764-4\_2}.

\bibitemdeclare{misc}{barr:beck-distributivity}
\bibitem{barr:beck-distributivity}
\bibinfo{author}{Michael Barr} (\bibinfo{year}{2005}):
  \emph{\bibinfo{title}{Beck Distributivity}}.
\newblock \urlprefix\url{ftp://ftp.math.mcgill.ca/barr/pdffiles/distlaw.pdf}.

\bibitemdeclare{incollection}{beck:distributivity}
\bibitem{beck:distributivity}
\bibinfo{author}{Jon Beck} (\bibinfo{year}{1969}):
  \emph{\bibinfo{title}{Distributive laws}}.
\newblock In: {\sl \bibinfo{booktitle}{Seminar on Triples and Categorical
  Homology Theory}}, {\sl \bibinfo{series}{Lecture Notes in
  Mathematics}}~\bibinfo{volume}{80}, \bibinfo{publisher}{Springer Berlin /
  Heidelberg}, pp. \bibinfo{pages}{119--140}, \doi{10.1007/BFb0083084}.

\bibitemdeclare{inproceedings}{DJPTFP:2011}
\bibitem{DJPTFP:2011}
\bibinfo{author}{Germ\'an~A. Delbianco}, \bibinfo{author}{Mauro Jaskelioff} \&
  \bibinfo{author}{Alberto Pardo} (\bibinfo{year}{2011}):
  \emph{\bibinfo{title}{{Applicative Shortcut Fusion}}}.
\newblock In: {\sl \bibinfo{booktitle}{Proceedings of the 12th International
  Symposium on Trends in Functional Programming}}, \bibinfo{address}{Madrid,
  Spain}.

\bibitemdeclare{incollection}{gambino04:wellfounded-trees}
\bibitem{gambino04:wellfounded-trees}
\bibinfo{author}{Nicola Gambino} \& \bibinfo{author}{Martin Hyland}
  (\bibinfo{year}{2004}): \emph{\bibinfo{title}{Wellfounded Trees and Dependent
  Polynomial Functors}}.
\newblock {\sl \bibinfo{series}{Lecture Notes in Computer Science}}
  \bibinfo{volume}{3085}, \bibinfo{publisher}{Springer Berlin / Heidelberg},
  pp. \bibinfo{pages}{210--225}, \doi{10.1007/978-3-540-24849-1\_14}.

\bibitemdeclare{unpublished}{backwards}
\bibitem{backwards}
\bibinfo{author}{Jeremy Gibbons} \& \bibinfo{author}{Richard Bird}
  (\bibinfo{year}{2011}): \emph{\bibinfo{title}{Effective Reasoning about
  Effectful Traversals}}.
\newblock
  \urlprefix\url{http://www.comlab.ox.ac.uk/jeremy.gibbons/publications/backwa%
rds.pdf}.
\newblock \bibinfo{note}{Submitted for publication}.

\bibitemdeclare{article}{Gibbons.Oliveira.Iterator}
\bibitem{Gibbons.Oliveira.Iterator}
\bibinfo{author}{Jeremy Gibbons} \& \bibinfo{author}{Bruno c. d.~s. Oliveira}
  (\bibinfo{year}{2009}): \emph{\bibinfo{title}{The essence of the iterator
  pattern}}.
\newblock {\sl \bibinfo{journal}{Journal of Functional Programming}}
  \bibinfo{volume}{19}, pp. \bibinfo{pages}{377--402},
  \doi{10.1017/S0956796809007291}.

\bibitemdeclare{inbook}{kelly74:elements-of-2categories}
\bibitem{kelly74:elements-of-2categories}
\bibinfo{author}{Gregory Kelly} \& \bibinfo{author}{Ross Street}
  (\bibinfo{year}{1974}): \emph{\bibinfo{title}{Review of the elements of
  2-categories}}, pp. \bibinfo{pages}{75--103}.
\newblock \bibinfo{volume}{420}, \bibinfo{publisher}{Springer Berlin /
  Heidelberg}, \doi{10.1007/BFb0063101}.

\bibitemdeclare{article}{mcbride08:applicative-programming}
\bibitem{mcbride08:applicative-programming}
\bibinfo{author}{Conor McBride} \& \bibinfo{author}{Ross Paterson}
  (\bibinfo{year}{2008}): \emph{\bibinfo{title}{Applicative programming with
  effects}}.
\newblock {\sl \bibinfo{journal}{Journal of Functional Programming}}
  \bibinfo{volume}{18}(\bibinfo{number}{01}), pp. \bibinfo{pages}{1--13},
  \doi{10.1017/S0956796807006326}.

\bibitemdeclare{article}{MoggiBJ99}
\bibitem{MoggiBJ99}
\bibinfo{author}{Eugenio Moggi}, \bibinfo{author}{Giana Bell{\`e}} \&
  \bibinfo{author}{C.~Barry Jay} (\bibinfo{year}{1999}):
  \emph{\bibinfo{title}{Monads, Shapely Functors and Traversals}}.
\newblock {\sl \bibinfo{journal}{Electronic Notes in Theoretical Computer
  Science}} \bibinfo{volume}{29}, pp. \bibinfo{pages}{187 -- 208},
  \doi{10.1016/S1571-0661(05)80316-0}.
\newblock \bibinfo{note}{CTCS '99, Conference on Category Theory and Computer
  Science}.

\bibitemdeclare{phdthesis}{Rypacek:2010}
\bibitem{Rypacek:2010}
\bibinfo{author}{Ond\v{r}ej Ryp\'a\v{c}ek} (\bibinfo{year}{2010}):
  \emph{\bibinfo{title}{Distributive Laws in Programming Structures}}.
\newblock Ph.D. thesis, \bibinfo{school}{University of Nottingham}.
\newblock \urlprefix\url{http://etheses.nottingham.ac.uk/1077/}.

\bibitemdeclare{misc}{skoda04:distributive-laws}
\bibitem{skoda04:distributive-laws}
\bibinfo{author}{Zoran Skoda} (\bibinfo{year}{2004}):
  \emph{\bibinfo{title}{Distributive laws for actions of monoidal categories}}.
\newblock \bibinfo{howpublished}{arXiv:math/0406310v2}.
\newblock \urlprefix\url{http://arxiv.org/abs/math/0406310v2}.

\end{thebibliography}

\end{document}